\documentclass[conference]{IEEEtran}

\usepackage{cite}
\usepackage{amsmath,amssymb,amsfonts,amsthm}
\usepackage[ruled, vlined, linesnumbered]{algorithm2e}
\usepackage{graphicx}
\usepackage{textcomp}
\usepackage{xcolor}

\usepackage{fixmath}

\usepackage[T1]{fontenc}
\usepackage[utf8x]{inputenc}
\usepackage[english]{babel}
\usepackage{nicefrac}

\usepackage{babel}

\ifCLASSOPTIONcompsoc
   \usepackage[caption=false,font=normalsize,labelfont=sf,textfont=sf]{subfig}
\else
   \usepackage[caption=false,font=footnotesize]{subfig}
\fi

\usepackage{hyperref} 

\usepackage{multirow}

\usepackage{versions}
\excludeversion{conference}
\includeversion{arxiv}

\usepackage{todonotes}
\usepackage[normalem]{ulem}

\graphicspath{{./images/}}

\newtheorem{definition}{Definition}
\newtheorem{theorem}{Theorem}
\newtheorem{lemma}{Lemma}

\begin{document}

\title{Attack-resistant Spanning Tree Construction in Route-Restricted Overlay Networks\\
}

\author{
   \IEEEauthorblockN{Martin Byrenheid}
   \IEEEauthorblockA{
      \textit{TU Dresden}\\
      martin.byrenheid@tu-dresden.de}
   \and
   \IEEEauthorblockN{Stefanie Roos}
   \IEEEauthorblockA{
      \textit{Delft University of Technology}\\
      s.roos@tudelft.nl}
   \and
   \IEEEauthorblockN{Thorsten Strufe}
   \IEEEauthorblockA{
      \textit{TU Dresden}\\
      thorsten.strufe@tu-dresden.de}
}

\maketitle

\begin{abstract}
Nodes in route-restricted overlays have an immutable set of neighbors, explicitly specified by their users.
Popular examples include payment networks such as the Lightning network as well as social overlays such as the Dark Freenet.
Routing algorithms are central to such overlays as they enable communication between nodes that are not directly connected.
Recent results show that algorithms based on spanning trees are the most promising provably efficient choice.
However, all suggested solutions fail to address how distributed spanning tree algorithms can deal with active denial of service attacks by malicious nodes.

In this work, we design a novel self-stabilizing spanning tree construction algorithm that utilizes cryptographic signatures and prove that it reduces the set of nodes affected by active attacks.
Our simulations substantiate this theoretical result with concrete values based on real-world data sets.
In particular, our results indicate that our algorithm reduces the number of affected nodes by up to 74\% compared to state-of-the-art attack-resistant spanning tree constructions.
\end{abstract}

\section{Introduction}
\noindent 
Payment or state channel networks like Lightning~\cite{poon2016lightning} are the most promising approach to scaling blockchains, i.e., enabling blockchain-based payment systems to process tens of thousands of transactions per second with nearly instant confirmation. 
Participants in such payment networks establish channels for trading assets such as digital coins. 
As establishing channels requires use of the blockchain, which is both time- and cost-intensive, only nodes that frequently trade with each other establish payment channels~\cite{dziembowski2018general}. 
All other payments pass from a sender to the receiver via multi-hop paths of channels.  
It is essential to find these paths in an effective, efficient, and privacy-preserving manner~\cite{roos2018settling}. 

Similarly, social overlays require finding paths from a peer to another in a network consisting only of connections between trusted pairs of nodes to realize scalable and privacy-preserving distributed services~\cite{roos2016anonymous,clarke2010private}. 

Both payment channel networks and social overlays hence share the need for a routing algorithm. 
A number of promising algorithms for both networks rely on Breadth-First-Search (BFS) spanning trees~\cite{malavolta2017silentwhispers,roos2016anonymous,roos2018settling}, as these permit finding shortest paths and achieve the most efficient communication.
{The underlying spanning tree construction algorithm determines the effectiveness, efficiency, and attack resilience of the routing. 
Resistance to attacks by malicious parties who aim to prevent the tree construction from converging towards a correct spanning tree is particularly important. Preventing the construction of a correct spanning tree results in routing failures and  hence constitutes a denial-of-service attack that undermines communication. Such attacks are realistic for both payment channel networks and social overlays.
For payment channel networks, adversarial parties may undermine the routing of payments to sabotage competing operators. 
Social overlays such as Freenet aim to protect communication from censorship~\cite{clarke2010private}. 
They clearly require attack resistance against participants aiming to execute censorship in the form of a denial-of-service attack.}   

{In the context of route-restricted overlays with potentially malicious participants, spanning tree algorithms have to fulfill three requirements:}
(1) enable efficient communication by providing short paths between {honest} nodes in the spanning tree, (2) efficiently adapt to changes of the network structure, 
and (3) maintain high availability in the presence of {malicious} nodes that deliberately deviate from the construction protocol in order keep the network from converging.
Yet, the existing work {on spanning tree-based routing} only evaluates the first two aspects jointly, leaving protection against malicious behavior out of scope {despite the likely existence of malicious parties in both payment channel networks and social overlays}. 

In this work, we focus on achieving all three requirements jointly, giving rise to two key contributions: 
\begin{itemize}
   \item We present a self-stabilizing algorithm for the computation of a BFS spanning tree that uses cryptographic signatures to check the integrity of statements about the distance to the root node. We prove that the fraction of nodes reaching a stable, non-compromised state is higher than in state-of-the-art protocols.
   \item We present results from an extensive simulation study based on real-world data sets. The results demonstrate that the construction of BFS spanning trees without cryptographic measures is highly vulnerable to attacks, even if the adversary establishes just a handful of connections to honest nodes. Furthermore, we show that our algorithm {substantially} raises the necessary number of such attack connections to mislead a comparable number of nodes.
\end{itemize}


\section{Related Work}
\label{sec:related-work}
{We review the existing work for routing in route-restricted overlays to show that the design of attack-resistant spanning trees is indeed the key problem to solve. Afterwards, we consider the existing work on attack-resistant spanning tree constructions, which we then improve upon in the following sections.}
\subsection{Routing in route-restricted overlays}
{We define an \emph{overlay network}, or just overlay, as a network between multiple logically connected nodes that communicate via a public infrastructure such as the Internet.
In \emph{route-restricted overlays}, the logical connections between nodes are explicitly managed by their respective users and hard or even impossible to adapt to create a topology that benefits routing.}
Apart from finding existing paths between nodes, routing algorithms have to be efficient and scalable with regard to delays for the delivery of messages, bandwidth and memory consumption to provide adequate service for large-scale peer-to-peer networks such as payment channel networks and social overlays.
Recent work~\cite{roos2016anonymous,malavolta2017silentwhispers,roos2018settling} underlines that only routing algorithms based on rooted spanning trees provide the necessary efficiency.
Other approaches either use expensive flooding for path discovery~\cite{isdal2011privacy} or setup virtual tunnels~\cite{vasserman2009,mittal2012,prihodko2016flare}, which, in face of network dynamics, require costly maintenance~\cite{roos2015impossibility}. 
{Alternatively, some payment channel networks of smaller size use source routing~\cite{poon2016lightning,sivaraman2018routing}, which requires that each node maintains a snapshot of the entire network. 
Source routing hence does not scale, as any change to the network has to be broadcast.} 

In the context of social overlays, Hoefer et al.~\cite{hoefer13greedy} suggested using greedy embeddings based on rooted spanning trees to enable efficient routing between nodes. 
The approach has later been extended to preserve the privacy of users and offer higher attack resistance~\cite{roos2016anonymous}. However, their adversarial model only considers the routing and not the construction of the underlying spanning tree, which is an orthogonal approach to the one taken in this paper. 

For payment channel networks, Malavolta et al.~\cite{malavolta2017silentwhispers} adapted Landmark Routing~\cite{tsuchiya1988landmark}, where a path between sender and receiver is determined through an intermediate node via the construction of a breadth-first-search tree rooted at the latter.
Roos et al. later on adapted the greedy embeddings to payment channel networks~\cite{roos2018settling}.
Both works aim to achieve efficiency and privacy and do not consider security.


It thus remains an open question, if and how such spanning trees can be constructed in route-restricted overlays with malicious participants. 


\subsection{Attack-resistant spanning tree construction}

In the context of self-stabilization, Dubois, Masuzawa, and Tixeuil proposed a BFS spanning tree algorithm and proved that this algorithm guarantees that all nodes, except those that are strictly closer to the adversary than to the root node, will eventually converge to a correct state~\cite{dubois2015maximum}.
While the algorithm by Dubois et al. offers provable attack resistance, it considers a computationally unbounded attacker.
Protecting against such a strong adversary disregards mechanisms such as digital signatures that can help to further decrease the number of affected nodes.

In the context of distance vector routing, which implicitly relies on BFS trees, Zapata and Asokan~\cite{zapata2002securing} proposed a protocol that utilizes hash chains to keep malicious nodes from lying about their distance from the root node.
Furthermore, their protocol employs cryptographic signatures to prevent attacks on the mechanism for the detection of routing loops.
Subsequently, Hu et al.~\cite{hu2003sead} proposed a protocol that uses hash chains both against attacks on the reported distance as well as against attacks on loop-detection, thus reducing computational overhead compared to digital signatures.
In contrast to the work of Dubois et al., both approaches assume a computationally bounded attacker. 
However, they do not provide a formal proof of their security guarantees. 
 
In summary, there exists no provably secure BFS tree construction algorithm under the assumption of a computationally bounded attacker. We expect that such an algorithm can provide protection to a larger set of nodes than the existing information theoretically secure algorithms.

\section{Model and Notation}
\label{sec:model}
\noindent We now formalize route-restricted overlays as well as the problem of computing a breadth-first-search tree in the context of self-stabilization. 

\subsection{System model}
\noindent We model a route-restricted overlay $S=(V,E)$ as a finite set $V$ of $n$ nodes and a set of bidirectional communication links $E \subset V \times V$.
For each node $u$, the set $N(u)=\{v \mid \{u,v\} \in E\}$ denotes the \emph{neighbors of $u$}.

We build upon the shared memory model where each pair of nodes $\{u,v\} \in E$ can communicate via shared registers $r_{uv}$ and $r_{vu}$, where $u$ is only allowed to write into $r_{uv}$ and read from $r_{vu}$.
We thus call $r_{uv}$ $u$'s \emph{output register} and $r_{vu}$ its \emph{input register}.

Please note that we use the shared memory model solely to simplify formal analysis, as it omits the modeling of message transmission.
We consider this to be reasonable, as we focus on malicious node behavior and neither link failures nor delays. 

\noindent For the computation of a BFS tree, every node $u$ holds the following elements:
\begin{itemize}
   \item $ID_u$, a fixed, globally unique ID from a set $\mathcal{ID}$,
   \item $level_u$, a non-negative integer variable denoting $u$'s current, assumed distance to the root,
   \item $pID_u$, a variable holding the $ID$ of the node that is currently considered parent, in other words, the neighbor of $u$ on the path to the root in the subgraph corresponding to the tree.
\end{itemize}
Furthermore, each communication register holds two values $ID$ and $level$ such that each output register of a node $u$ holds $u$'s $ID$ as well as its current $level$-value.
Each input register $r_{vu}$ of a node $u$ accordingly holds $u$'s current view of $v$'s $ID$ and $level$-value.
In the following, we denote the set $N_{min}(u) = \{ v \in N(u) | \forall n \in N(u): level_v \leq level_n\}$ as \emph{minimal neighbors} of $u$. Parent nodes are always minimal neighbors in BFS spanning trees. 

We refer to the values currently held by the $level$- and $pID$-variable of a node $u$ as well as the register contents, at one point in time, as the \emph{state} of $u$.
The state of $u$ is said to be \emph{legitimate} if it fulfills Def.~\ref{def:legit-state}.
\begin{definition}
   \emph{(Legitimate state)}
   Let $S=(V,E)$ be a route-restricted overlay with a distinguished root node $l \in V$ with $ID$-value $ID_L \in \mathcal{ID}$.
   The state of a node $u$ whose minimal neighbors have level $l_{min}$ is called \emph{legitimate} if it fulfills the following conditions:
   \begin{enumerate}
      \item $level_u = 0$ iff $ID_u = ID_L$
      \item $level_u = l_{min} + 1$ if $ID_u \neq ID_L$
      \item $pID_u = ID_u$ iff $ID_u = ID_L$
      \item $\exists v_{min} \in N_{min}(u): pID_u = ID_{v_{min}}$ if $ID_u \neq ID_L$
   \end{enumerate}
   \label{def:legit-state}
\end{definition}

\subsection{Adversary model}
\label{sec:attacker}
\noindent In this work, we consider adversaries who aim to perform large-scale denial of service attacks.
For payment networks, they might be competing payment network operators who want to attract more users by rendering other networks unusable.
For social overlays, the adversary might aim to weaken the privacy~\cite{borisov2007denial} or degrade utility so that users move to communication services with weaker privacy protection.

Allowing multiple adversaries to act in concert strictly increases their power.
We hence assume a single, collective adversary who controls a set $B$ of \emph{malicious} (or \emph{adversarial}) nodes and is able to set up a bounded number of connections between these malicious and \emph{honest} nodes $H$. 
The motivation for these bounds is the difficulty of large-scale social engineering that will only be successful for a subset of participants.

During an attack, each malicious node may report incorrect data to the adjacent honest nodes in order to keep them from reaching or remaining in a legitimate state.
Thus, malicious nodes may set their output registers arbitrarily and report different $ID$- and $level$-values to different neighbors.

However, we assume that the adversary does not know all honest nodes and their internal connections a priori.
Hence, he cannot choose which nodes will be malicious or which nodes will connect to malicious nodes. 
Given that social overlay and payment networks are large-scale and dynamic distributed systems with participants from a multitude of countries, we consider this assumption to be realistic.

For all practical purposes, the Dolev-Yao model, which assumes an adversary who is limited to polynomial-time attacks -- and hence unable to break secure cryptographic primitives -- has been accepted as realistic\cite{dolev83security}.
Hence, we aim for algorithms that protect against adversaries that are polynomially bounded.

\subsection{Formalization of resilience and performance}
We formalize the attack resistance of a spanning tree construction protocol via the concept of topology-aware (TA) strict stabilization~\cite{dubois2015maximum}. 
To do so, we express the state of every node in the overlay at one point in time as a \emph{configuration} $\gamma$.

Following the idea of self-stabilization, we consider that every node starts in an arbitrary state.
Thus, nodes may change their state over time to reach a legitimate state. 
The sequence of configurations $\gamma_0,\gamma_1,\dots$ is called a \emph{computation} $\Gamma$.
The transition from $\gamma_{t}$ to $\gamma_{t+1}$ is called a \emph{step} and corresponds to at least one node processing the data in its input register and writing corresponding data into its output register.

Note that self-stabilizing algorithms never terminate but repeatedly update their state and communication registers.
However, a node executing a step may not actually change the values of its variables or output registers {(e.g., because its current state is legitimate)}. 

\paragraph{Network dynamics}
Route-restricted overlays are dynamic: nodes may join and leave the system, connections between nodes are established or torn down over time.
An overlay $S=(V,E)$ changes into an overlay $S'=(V',E')$ with a potentially different network size as a consequence of such events.
According to literature, we call such changes \emph{churn events}. 
To account for the fact that computations are defined for a fixed system $S$, a churn event interrupts a computation on $S$ and starts a new computation on $S'$.

At the beginning of the new computation, all nodes in $V \cap V'$ have the same state as at the end of the computation on $S$, reflecting the fact that they cannot detect the change until they read from their registers. 
The remaining nodes in $V'$ may start in an arbitrary initial state.
In route-restricted networks, the initial state includes information about the register of neighbors, which the new node will eventually write to. 

\paragraph{Containment of attacks}
TA strict stabilization for a set $S_B \subset H$ of honest nodes denotes that every honest node $u$ except those in the set $S_B$ eventually reaches and remains in a legitimate state.
We call the set $S_B$ the \emph{containment area of $S$}, because $S_B$ (also called \emph{lost nodes}) represents the part of the network where the adversary can keep the state of nodes from converging, whereas all nodes outside of $S_B$ (called \emph{safe nodes}) will eventually reach and remain in a legitimate state.

We now formalize the concept of a node having only honest ancestors on its path to the root. 
\begin{definition}
   \emph{(Root-directed path)}
   Given a route-restricted overlay $S$ and a configuration $\gamma$, the root-directed path $P_u$ of a node $u$ is a finite sequence $v_1,v_2,\dots,v_{n+1}$ of nodes in a legitimate state such that $v_{n+1}=u$ and $pID_{v_{i+1}} = ID_{v_{i}}$ for all $1 \leq i \leq n$ and either $pID_{v_1}=ID_{v_1}$ (the legitimate root) or $v_1$ is a malicious node.
   We call $u$ \emph{ill-directed} if $v_i$ is malicious for any $1 \leq i \leq n$ and \emph{well-directed} otherwise.
   \label{def:root-path}
\end{definition}
As long as a node is ill-directed, it is subject to changes in the $level$-value reported by the adversarial node on its root-directed path.
Thus, it is not guaranteed to remain in a legitimate state.
However, an ill-directed node is not inherently a lost node, because it might eventually become well-directed as the execution proceeds.

We express the situation that a node's state has converged and remains unaffected by attacks as follows:
\begin{definition}
   \emph{(Stable state)}
   The state of a node $u$ is said to be \emph{stable} if it is legitimate and $u$ never changes its $level_u$- and $pID_u$-variable as long as no churn event occurs. In particular, actions performed by malicious nodes do not affect $u$. 
   A configuration $\gamma$ is called \emph{$S_B$-stable} if the state of every node in $V\setminus S_B$ is stable.
   \label{def:legit}
\end{definition}

We define a $S_B$-topology-aware-strictly-stabilizing ($S_B$-TA-strictly-stabilizing) algorithm as follows: 
\begin{definition}
   \emph{($S_B$-TA-strictly-stabilizing algorithm)}
   A distributed algorithm $\mathcal{A}$ is $S_B$-TA-strictly-stabilizing if and only if starting from an arbitrary configuration, every execution contains a $S_B$-stable configuration.
   \label{def:stabilizing}
\end{definition}

\paragraph{Time complexity}
To be able to reason about the time complexity that a distributed algorithm requires to reach a legitimate state, we use the concept of \emph{asynchronous rounds}.
The first \emph{asynchronous round} of a computation $\Gamma$ is the shortest prefix $\Gamma'$ of $\Gamma$ such that each node has read from and wrote to all of its registers at least once. 
The second asynchronous round then is the first asynchronous round of the computation following $\Gamma'$ and so on. In other words, the length of an asynchronous round corresponds to the maximum amount of time needed for the slowest node (regarding computational speed) to process its inputs and write the corresponding outputs.

\section{Signature-based computation of BFS trees}
\label{sec:computing}
\label{sec:attestation}
\noindent
The state-of-the-art algorithm for the construction of BFS trees proposed by Dubois et al.\cite{dubois2015maximum} ensures that all honest nodes whose distance from the closest malicious node is higher or equal than their distance from the root will eventually reach a stable state. 
As the set of nodes that do not reach a stable state is often quite large for this algorithm, we investigate algorithms that achieve a higher number of stable nodes. In contract to previous work, we assume our adversary to be computationally bounded.

In our design, each node $u$ holds a public/private key pair $p_u,s_u$ of an asymmetric cryptosystem.
The public key $p_u$ of each node $u$ is stored in the $ID$-register and the secret $s_u$ is stored in a new register called $secret_u$.
The given leader ID $ID_L$ then is the public key of the corresponding root node, implicitly choosing it as leader.
Nodes do not require global knowledge of all other nodes' keys.

\paragraph{Assumptions}
Four assumptions underlie our design:
\begin{itemize}
   \item There is an honest root node whose key is known to all nodes (e.g., bank in a payment network~\cite{malavolta2017silentwhispers}).
   \item The clocks of any pair of nodes differ at most by a globally known constant $\varDelta_C$.
   \item The time needed for one iteration of each node's main loop is bounded by a globally known constant $\varDelta_E$.
   \item The delay needed until a value written into an output register is available in the corresponding input register is bounded by a globally known constant $\varDelta_D$.
\end{itemize}
The first assumption is in accordance with the existing literature on tree-based routing in route-restricted overlays~\cite{malavolta2017silentwhispers,roos2018settling,roos2016anonymous}.
The remaining assumptions allow us to compute expiration times for the data contained in the input register of each node, thus keeping malicious nodes from reporting outdated values obtained in previous computations.

\paragraph{Level attestation}
To keep malicious nodes from lying about their distance to the root, we add a $levelAtt$-variable to each node $u$, which holds a finite sequence $P=(p_1, t_1, sig_1),(p_2, t_2, sig_2),\dots,(p_n, t_n,sig_n)$ of tuples called a \emph{level attestation}.
The elements $p_i$, $t_i$, and $sig_i$ denote a public key, a timestamp, and a cryptographic signature, respectively.
We say that such a sequence is \emph{valid for node $u$ at time $t$} if the following conditions are satisfied:
\begin{enumerate}
   \item $p_1 = ID_L$,
   \item $\forall i \in \{1,..,n\}:$ $t - t_i \leq \varDelta_C + (\varDelta_D + \varDelta_E)(n-i+1)$,
   \item $\forall i \in \{1,..,n-1\}:$ $sig_i$ is a signature over $p_{i+1} || t_i$ that is valid for $p_i$,
   \item $sig_{n}$ holds a signature over $ID_u || t_n$ that is valid for $p_n$,
\end{enumerate}
where $a||b$ denotes the concatenation of $a$ and $b$.

Condition (1) ensures that the first tuple of the attestation has indeed been generated by the root node.
Condition (2) ensures that adversarial nodes cannot use obsolete attestations (e.g., from an earlier computation) forever.
Condition (3) and (4) ensure that the signatures are computed correctly.

\paragraph{Link signatures}
Additional to the level attestation, each node assigns a randomly chosen \emph{neighbor ID} $nID_v$ to each neighbor $v$ once in the beginning of the algorithm.
During the computation, every honest node tells each neighbor its respective neighbor ID.
Whenever a neighbor of a node $u$ transmits a new level attestation, it also has to send a corresponding \emph{neighbor signature} that includes its neighbor ID assigned by $u$.
Given a valid level attestation $P$ with the last element $(p, t, sig)$ and a cryptographic hash function $h$, a neighbor signature $s$ is \emph{valid for node $u$ and neighbor $v$} if $s$ is a valid signature over $nID_v || h(P)$ for $p$.
This addendum keeps malicious nodes from sending a shortened version of an attestation received by an honest neighbor.

\paragraph{Adaptive neighbor preference}
To ensure stabilization in the case that a node has multiple neighbors that are minimal according to Def.~\ref{def:legit-state}, each node $u$ assigns a unique number between $0$ and $|N(u)|-1$ to each neighbor and chooses the minimal neighbor with the lowest number as parent.
The number of the current parent is kept in a variable $prnt$.
As the preferred neighbor may be ill-directed, the algorithm of Dubois et al.~\cite{dubois2015maximum} adaptively changes which neighbor will be preferred whenever a node changes its parent.
We implemented this strategy as follows: We add an offset counter $i_{start} \in \{0,..,|N(u)|-1\}$ such that $u$ traverses its neighbors from $i_{start}$ to $(|N(u)|-1) + i_{start} \mod |N(u)|$.
Whenever a node $u$ changes its parent from the neighbor with number $prnt$ to a neighbor with a number $prnt'$ that, counting from $i_{start}$ with wraparound, comes after $prnt$, then $u$ will set $i_{start}$ to $prnt'$, thus favoring $prnt'$ over $prnt$ in the future.
To compare nodes' positions $a$ and $b$ with regard to $i_{start}$, we say that $a \prec_{i_{start}} b$ if either i) $i_{start} \leq a < b$, ii) $b < i_{start} \leq a$ or iii) $a < b < i_{start}$.
Informally, $a \prec_{i_{start}} b$ indicates that $b$ be will be reached later than $a$ when counting from $i_{start}$ modulo $|N(u)|$.

\paragraph{Spanning Tree Algorithm}
\newlength{\oldfloatsep} \setlength{\oldfloatsep}{\textfloatsep}
\setlength{\textfloatsep}{0pt} 
\begin{algorithm}
\small
\DontPrintSemicolon
\While{true} {
   \ForEach{$i$ in $N(u)$} {
      $lr_{iu} := \text{\textbf{read}}(r_{iu})$\;
   } $ts := \mathit{getCurrentTime()}$\;
   $i_{start} := i_{start} \mod |N(u)|$\;\label{line:istart}
   \If{$ID = ID_L$}{
      $pID := ID$\;\label{line:root1}
      $level := 0$\;\label{line:root2}
      $levelAtt := nil$\;
   }
   \Else{
       $parentFound := false$\; \label{line:process1}
    $N_{valid} := \{ i \in N(u) \mid \mathit{isValidAtt(lr_{iu}.levelAtt, lr_{iu}.level+1)} \wedge \mathit{isValidLink(lr_{iu}.levelAtt, lr_{iu}.sig_{adj})\}}$\; \label{line:valid}
          $level := \min \{lr_{iu}.level \mid i \in N_{valid}\} + 1$\; \label{line:min}
          \ForEach{$i$ in $1..|N(u)|$} {
          $j := i + i_{start} \mod |N(u)|$\; \label{line:j}
          \If{not $parentFound$ and $N(j) \in N_{valid}$ and $level = lr_{ju}.level + 1$} {
            
             \If{$prnt \prec_{i_{start}}j$} {\label{line:check}
                $i_{start} := j$\;
            }
            $prnt := j$\;
            $pID := lr_{ju}.ID$\;
            $levelAtt := lr_{ju}.levelAtt$\;
            $parentFound := true$\; \label{line:process2}
          }
       }
   }
   \ForEach{$i$ in $N(u)$} {
      $sig_{lvl} := sign(lr_{iu}.ID || ts)$\; \label{line:write1}
      $exAtt := append(levelAtt, (ID, ts, sig_{lvl}))$\;\label{line:append}
      $sig_{adj} := sign(lr_{iu}.nID || h(exAtt))$\;
      $\text{\textbf{write}}(r_{ui}) := (ID, level, exAtt, nID_i, sig_{adj})$\; \label{line:write2}
   }
}
\caption{Attestation-based spanning tree on node $u$}
   \label{alg:bfstree}
\normalsize
\end{algorithm}
Algorithm~\ref{alg:bfstree} displays the pseudocode for our spanning tree construction algorithm: Each output register of every node $u$ holds 5 elements, namely the \mbox{$ID$-} and $level$-value of $u$ as well as the $levelAtt$- and $nID$-value together with the neighbor signature $sig_{adj}$ for the corresponding neighbor.
The algorithm leverages the following cryptographic functions: The $sign$-function uses the key stored in the $secret$-register to compute a signature $sig$. The function $h$ is a cryptographic hash function. 

Every node periodically reads the content of each input register, processes the content, and writes corresponding outputs to output registers.
The leader node first ensures that its $pID$- and $level$-value are set correctly (Line~\ref{line:root1}--\ref{line:root2}). 
Subsequently, it generates a level attestation for each neighbor and writes its own $ID$ and $level$-value together with the respective $nID$-value, level attestation, and neighbor signature into the corresponding output register (Line~\ref{line:write1}--\ref{line:write2}).
Because the $levelAtt$-variable is set to $nil$, the $\mathit{append}$-operation in Line~\ref{line:append} just returns its second argument.

During the processing stage (Line~\ref{line:process1}--\ref{line:process2}), an honest non-leader node recomputes its current $pID$-, $prnt$-, $level$- and $levelAtt$-value. It first checks the validity of the received level attestations and neighbor signatures and computes the set of valid neighbors in Line~\ref{line:valid}.
The $isValidAtt$-function checks whether a given level attestation is valid, as defined above.
If the given level attestation is valid, $isValidAtt$ further checks whether the length of the attestation equals the given level value and returns false in case of a mismatch.
Given this check succeeds, the $isValidLink$-function checks if a given $sig_{adj}$-value is valid for the corresponding neighbor.
If a parent node with a valid level attestation has been chosen, the node first checks if its previous parent became either non-minimal or its attestation became invalid and if so, sets $i_{start}$ to $j$.
It is possible that $prnt$ might hold a value larger than $|N(u)|-1$ (e.g. because its former parent had this number and left the overlay).
$prnt$ will then be set to $j$ that holds a value from the range $\{0,..,|N(u)|-1\}$ (Line~\ref{line:j}).
Afterwards, it sets its $prnt$-, $pID$- and $levelAtt$-value accordingly.
Finally, the node computes the corresponding level attestation for each neighbor and writes it into the respective output register (Line~\ref{line:write1}--\ref{line:write2}).

\setlength{\textfloatsep}{\oldfloatsep}

\section{Analysis}
\label{sec:proofs}
We prove that, given an honest root node $r$, Algorithm~\ref{alg:bfstree} is $S'_B$-TA-strictly-stabilizing with 
\begin{equation}
   S'_B = \{ u \in H \mid \exists b \in B: d^B_{min}+d_S(b,u)-1 \leq d_S(r,u)\}
   \label{form:SB}
\end{equation}
where $d^B_{min}=\min_{b \in B} d_S(r,b)$.
The ``-1'' stems from the fact that a malicious node can copy the outputs of an honest neighbor into its output registers (hence pretending to be its own predecessor), thus avoiding the need to append an attestation tuple and hence increase its maximum level.

Furthermore, let $d^H_S(u,v)$ denote the length of the shortest path between $u$ and $v$ in $S$ that does not contain a malicious node.
If no such path exists, we set $d^H_S(u,v)=\infty$.
If malicious nodes repeatedly change their outputs in order to de-stabilize honest nodes, we show that our algorithm guarantees that all nodes in the set
\begin{equation}
   S'_L = \{ u \in H \mid \exists b \in B: d^B_{min}+d_S(b,u)-1 < d^H_S(r,u)\}
   \label{form:SL}
\end{equation}
eventually reach a stable state. 
Informally, we show that $S'_L \subset S'_B$ is the containment area for an adversary that focuses on disrupting convergence by changing its behavior. 
However, for an arbitrary adversary aiming to maximize the fraction of ill-directed nodes, we achieve only a smaller containment area of $S'_B$.

Since the system starts in an arbitrary state, a malicious node may initially hold a level attestation that is valid but for which no corresponding path in the overlay exists.
We hence say that a level attestation $(p_1, t_1, sig_1), \ldots , (p_n, t_n, sig_n)$ is \emph{consistent for node $u$} if it is invalid or if there exists a path $v_1, \ldots, v_n$ in the system such that (1) $p_i$ is the public key of $v_i$ for all $1 \leq i \leq n$ and (2) $u$ either is a neighbor of $v_n$ or both $u$ and $v_n$ are neighbors of a malicious node $b$.
Otherwise, we say that the attestation is \emph{inconsistent}.
A configuration is called consistent if the $levelAtt$-values as well as the in- and output-registers of all nodes only contain consistent level attestations.

In the following, we assume that at the beginning of a computation at time $t$, all timestamps of every inconsistent attestation are at most $t + \varDelta_C$.
We consider this to be reasonable since $t + \varDelta_C$ is the highest value that a honest node (including the root) may use as timestamp and thus a malicious node cannot have a valid attestation with a higher timestamp from a previous computation.
As a consequence, every inconsistent attestation of length $n$ becomes invalid after at most $\varDelta_C + (\varDelta_D + \varDelta_E)n$ time units. 
So, every route-restricted overlay $S$ with diameter $diam(S)$ reaches a consistent configuration after at most $\varDelta_C + (\varDelta_D + \varDelta_E)diam(S)$ time units.

\subsection{Proof of $S'_B$-TA-strict stabilization}
We start the actual proof by establishing key properties of level attestation to later leverage in the proof.
In a nutshell, malicious nodes can only influence keys that are used after the $d_{min}^B$-th element of a valid and consistent level attestation $P$ but before the $|P|-d^B_{u,min}$-th element with $d^B_{u,min}=\min_{b \in B}\{d_S(u,b)\}$.
Based on this result, we can then show that a node is well-directed if their $levelAtt$-value is of length less than $d^B_{min}+d^B_{u,min}-1$.
Convergence to a stable state for all nodes in $S'_B$ follows from the fact that the system at some point reaches a state when these nodes have a valid and consistent $levelAtt$-value with minimal levels and hence will not change parents anymore.

\begin{lemma}
   Let $P=(p_1, t_1, sig_1), \ldots , (p_n, t_n, sig_n)$ be a level attestation.  
   Consider a node $u$ such that $sig_n$ is a signature over $ID_u||t_n$.
   At time $t$, we have $t-t_i\leq \Delta_C + (\Delta_D + \Delta_E) \cdot (n-i+1)$ for all $1\leq i \leq n$ and the computation has started at least $\Delta_C + (\Delta_D + \Delta_E) \cdot n$ time units before, so that $P$ is consistent for $u$.
If P is valid, then the following two statements hold:
   \begin{enumerate}
   \item For $j \leq d^B_{min}$, $p_j$ is the public key of an honest node $v$ and $d_S(v,r) < j$.
   \item For $j > n-d^B_{min, u}+1$, $p_j$ is the public key of an honest node $v$ and $d_S(v,u) \leq n-j+1$.
   \end{enumerate}   
   \label{lemma:position}
\end{lemma}

\begin{conference}
   Lemma~\ref{lemma:position}  follows by induction on $j$ and the full proof  can be found in the extended version of the paper~\cite{byrenheid2019attack}. 
\end{conference}
\begin{arxiv}
\begin{proof}
We show the first claim by induction on $j$.
As $p_1$ always needs to be the public key of the leader and the leader is honest by assumption, the claim holds for $j=1$.
Let $1 < j \leq d^B_{min}$ and assume the claim holds for $j-1$. Then $sig_{j-1}$ is a signature over $p_j||t_{j-1}$ using the secret key $s_{j-1}$ associated with $p_{j-1}$.
By induction hypothesis, $p_{j-1}$ is the public key of an honest node $w$ with distance $d_S(w,r) < j -1 \leq d^B_{min}-1$. $d_S(w,r) < d^B_{min}-1$ implies that $w$ has only honest neighbors, which only write their own keys to its output register for $w$ to sign.

Furthermore, because $w$ itself is honest, $w$ only signs keys and timestamps that it reads from its input registers. 
Thus, for $p_j||t_{j-1}$ to be signed by $w$, $p_j$ needs to be the key of an honest neighbor $v$ of $w$.
Given that $w$'s distance to the root is less than $j-1$ by induction hypothesis, we also have $d_S(v,r)\leq d_S(w,r)+1<j$.
This proves the first claim.
Similarly, we show the second claim by induction on $j'=n-j+1$.
Note that if $d^B_{min, u}=1$, i.e., $u$ is the neighbor of a malicious node, then there is nothing to show as there is no $p_j$ such that $j > n-d^B_{min, u}+1$. So, we assume $d^B_{min, u}>1$.
For $j'=1$, we only have to consider the key $p_n$.
As $u$ is honest, it only writes its own key into output registers to be signed by neighbors.
If $d^B_{min, u}>1$, all of $u$'s neighbors are honest.
They would hence only sign $u$'s key concatenated with a timestamp with their own, meaning that any node $v$ with public key $p_n$ is indeed an honest node and $d_S(v,u)=1$. 
Consider $1<j'<d^B_{min, u}$ and assume the claim holds for $j'-1$.
Hence, $p_{n-(j'-1)+1}$ is the public key of an honest node $w$ with $d_S(w,u)\leq j'-1$.
$w$ writes its public key and a timestamp to the registers that will be read by its neighbors. As $j'-1 < d^B_{min, u} -1$, these neighbors are honest and will sign the key and timestamp with their own keys.
Hence, any public key $p_{n-j'+1}$ whose corresponding secret key has been used to sign $p_{n-(j'-1)+1}||t_{n-(j'-1)}$ belongs to an honest neighbor $v$ of $w$ with $d_S(v,u)\leq d_S(w,u)+1=j'$. 
So, the second claim follows by induction as well.
\end{proof}
\end{arxiv}
\begin{lemma}
   Let the computation have started a least $\Delta_C + (\Delta_D + \Delta_E) \cdot n$ time units before and $u\in V\setminus S'_B$ be a node with a valid $levelAtt$-value of length $n < d^B_{min}+d^B_{u,min}-1$. Then $u$ is well-directed. 
   \label{lemma:lengthLevelAtt}
\end{lemma}
\begin{proof}
Because $\Delta_C + (\Delta_D + \Delta_E) \cdot n$ time units have passed, the $levelAtt$-value of $u$ is also consistent.
By Lemma~\ref{lemma:position}, the first $d^B_{min}$ public keys have to belong to honest nodes and the last $d^B_{u,min}-1$ keys have to belong to honest nodes as well. Hence, if $n < d^B_{min}+d^B_{u,min}-1$, all keys $p_j$ have to belong to an honest node $v_j$ for $1\leq j \leq n$. Set $v_{n+1}=u$.

$u$ can only be ill-directed if at least one $v_j$ has their $pID$-value set to a key provided by a malicious node.
First, consider the case that $j < d^B_{min}$.
By Lemma~\ref{lemma:position}, $d_S(v_j,r)<d^B_{min}-1$, meaning that $v_j$ only has honest neighbors.
Honest nodes only write their own keys in the register of their neighbors, so that $v_j$ can hence only set its $pID$-value to one of their keys.
Now, consider $j > d^B_{min}$, i.e., $n-j+1< n- d^B_{min}+1 \leq d^B_{u,min}-1$.
According to Lemma~\ref{lemma:position}, $d_S(v_j, u)\leq n -j +1 < d^B_{u,min}-1$. 
Again, $v_j$ has only honest neighbors and can hence only set its $pID$-value to one of their keys.

It remains to consider the case $j=d^B_{min}$.
By the first part of the proof, $v_j$ is the only node that can have malicious neighbors.
Assume that $v_j$ has set its $pID$ to a malicious neighbor $b$.
For $u$'s $levelAtt$ to correspond to a valid attestation, $v_{j-1}$ has to sign $p_j||t_{j-1}$ resulting in $sig_{j-1}$, append $(p_{j-1}, t_{j-1}, sig_{j-1})$ to the attestation, and write the attestation to the register corresponding to the neighbor that wrote $p_j$ to the register. 
Because $v_{j-1}$ has only honest neighbors, the respective neighbor has to be $v_j$, the only honest node that would claim $p_j$ as its key.
So, for $u$'s $levelAtt$-value to include $(p_{j-1}, t_{j-1}, sig_{j-1})$, $v_j$ must have read the register and disseminated $(p_{j-1}, t_{j-1}, sig_{j-1})$ as part of a level attestation.
Consequently, $v_j$ is aware that $v_{j-1}$ offers a root-directed path of supposed length $j-2\leq d^B_{min}-2$.
For $v_j$ to choose a different parent, $b$ has to produce a valid attestation $\widetilde{P}=(\widetilde{p}_1, \widetilde{t}_1,\widetilde{sig}_1), \ldots , (\widetilde{p}_l, \widetilde{t}_l, \widetilde{sig}_l)$ with $l\leq j-1$ and $\widetilde{sig}_l$ being a signature over $ID_{v_j}||\widetilde{t}_l$.  Furthermore, $b$ has to ensure that the $isValidLink$-function returns $true$. The neighbor-related signature has to be signed by the secret key $\widetilde{s}_l$ corresponding to $\widetilde{p}_l$. As $b$ can not forge signatures, $\widetilde{P}$ has to be a (potentially shortened) attestation that $b$ has read from one of its input registers. For such an attestation, $\widetilde{p}_{l}$ belongs to an honest node  $w$ at distance at most $l-1$ from the root by Lemma~\ref{lemma:position}. Due to $d_S^H(w,r)\leq l-1<d^B_{min}-1$, $w$ has no malicious neighbors. By Algorithm~\ref{alg:bfstree}, $w$ only writes signatures over $nID_w || h(L)$ for some $L$ to registers of neighbors. Being honest, these neighbors do not disseminate the respective signatures. As a consequence, $b$ can not obtain the required neighbor signature and hence $v_j$ does not 
accept any attestation from $b$ as its $levelAtt$-value. 

In summary, none of the nodes $v_j$ has its $pID$-value set to a key provided by a malicious node and hence $u$ is indeed well-directed.  
\end{proof}
\begin{theorem}
   Given any route-restricted overlay $S$ with diameter $diam(S)$, a computation of Algorithm~\ref{alg:bfstree} starting from an arbitrary configuration reaches a consistent configuration after at most $\Delta_C + (\Delta_D + \Delta_E) \cdot diam(S)$ time units.
   Furthermore the computation will reach a $S'_B$-stable configuration $\gamma^*$ within at most $diam(S)+1$ additional asynchronous rounds. Thus, Algorithm~\ref{alg:bfstree} is $S'_B$-TA-strictly-stabilizing.
   \label{thm:att-selfStab}
\end{theorem}
\begin{proof}
   Arrival at a consistent configuration follows from the assumption that the timestamps of every inconsistent level attestation do not exceed the starting time of the computation by more than $\varDelta_C$ time units, as explained at the beginning of this section.
   To prove the subsequent convergence to a $S'_B$-stable configuration, we first show that after $l+1$ rounds, all nodes $u\in V\setminus S'_B$ within distance $l$ of the root are well-directed and have valid $levelAtt$-values of length $l$. The properties from Definition~\ref{def:legit-state} follow. Last, we show that these nodes remain well-directed.

   After the first round, the root has written its information to all registers. After the second round, the neighbors of the root have processed these registers. Hence, each such neighbor $u$ will set its $levelAtt$-value to a valid attestation of length 1. If $u\in V\setminus S'_B$, the distance $d_S(u,b)\geq 2$ for any malicious node $b$ and hence by Lemma~\ref{lemma:lengthLevelAtt}, $u$ is well-directed. 
So, the claim holds for $l=1$.

Assume the claim holds for $l$, i.e., after $l+1$ rounds, all nodes $v\in V\setminus S'_B$ within distance $l$ of the root are well-directed and have valid $levelAtt$-values of $l$. They know the IDs their neighbors have assigned to them as $l>1$ indicates that they have read it from the register at least once. As a consequence, they can construct a valid attestation of length $l+1$ for each neighbor $w$ as well as the necessary signature over the neighbor ID $nID_w$. They write this information to the register $r_{vw}$.  After $l+1$ rounds, any  node $u\in V\setminus S'_B$ at distance $l+1$ from the root has read the register corresponding to its neighbors at distance $l$ to the root. As a consequence, $u$'s $levelAtt$-value is of length  $l+1$. As $u\in V\setminus S'_B$, Lemma~\ref{lemma:lengthLevelAtt} shows that $u$ is well-directed. 
It follows by induction that within $diam(S)$ rounds, all nodes $u\in V\setminus S'_B$ are well-directed. 

It remains to prove that the nodes in $V\setminus S'_B$ remain well-directed. To become ill-directed, a node  has to change its $pID$-value. Let $u$ be the first node to change its $pID$-value. According to Algorithm~\ref{alg:bfstree}, $u$ selects the parent from those neighbors that provide the shortest valid attestation and a valid neighbor signature. By assumption, $u$ breaks ties consistently, meaning $u$ only changes its parent if either i) $u$'s previous parent does not provide any valid attestation of the shortest length or provides an invalid neighbor signature, or ii) a neighbor that is not the current parent writes an attestation of a shorter length than $u$'s $levelAtt$-value to its register and the content of the register passes the two validity checks.

In order to conclude that neither i) or ii) are possible, consider the following:  
Let $v$ be $u$'s parent and note that $v\in V\setminus S'_B$ by the definition of $S'_B$ as $d_S(v,r)=d_S(u,r)-1$ and $d_S(v,b)\geq d_S(u,b)-1$ for all malicious nodes $b$. It follows recursively that all nodes on a root-directed path of $u$ are in $V\setminus S'_B$. 
Case i) would imply that a node on the root-directed path changed its parent, as honest nodes do not write invalid attestations or neighbor signatures to registers. However, such a parent change contradicts the definition of $u$ as the first node in $V\setminus S'_B$ to change its parent.  
If case ii) holds, by Lemma~\ref{lemma:position}, $u$ has to be well-directed after its parent change. Hence, its new parent $w$ is an honest node. By the above, $w$ and all nodes on the new root-directed path are in  $V\setminus S'_B$ and at least one of them has to have changed its parent for $w$ to write an attestation of a different length. Again, such a change in parent is a contradiction to the definition of $u$.  
Consequently, nodes $u\in V\setminus S'_B$ do not change their $pID$-value for the rest of the computation and remain well-directed.
\end{proof}

\subsection{Proof of stabilization for $S'_L$ under attacks}
Building upon Theorem~\ref{thm:att-selfStab}, we now show that under an attacker that frequently changes the output values of its nodes, all nodes $u$ with $d^B_{min}-d^H_S(b,u)-1 = d^H_S(r,u)$ eventually reach a stable state as well.
Our result requires the concept of a $S_B$-disturbance, a concept similar to  Dubois et al.~\cite{dubois2015maximum}'s $S_B$-disruption. 
\begin{definition}
   \emph{($S_B$-disturbance)}
   Two consecutive configurations $\gamma_0$ and $\gamma_1$ are a \emph{$S_B$-disturbance} if at least one node $u \in V\setminus S_B$ changes its $level_u$- or $pID$-variable. 
   \label{def:disturbance}
\end{definition}
In contrast to a $S_B$-disruption, a $S_B$-disturbance does not assume that all nodes in $V\setminus S_B$ have a legitimate state.

\begin{theorem}
   \label{thm:SL}
   Given any route-restricted overlay $S$ with diameter $diam(S)$ and $deg_{sum}=\sum_{u \in S'_B\setminus S'_L} |N(u)|$, a computation of Algorithm~\ref{alg:bfstree} starting from an arbitrary configuration reaches a $S'_B$-stable configuration $\gamma^+$ within at most $\Delta_C + (\Delta_D + \Delta_E) \cdot diam(S)$ time units plus at most $diam(S)+1$ asynchronous rounds.
   After reaching the configuration $\gamma^+$,  $S$ will reach a $S'_L$-stable configuration within at most ($2deg_{sum}-|S'_B\setminus S'_L|$) $S'_L$-disturbances.
\end{theorem}
\begin{conference}
   We only present a sketch of the proof here and refer the reader to the full version of the paper~\cite{byrenheid2019attack}.
   The $S'_B$-stability after at most $\varDelta_C + (\varDelta_D + \varDelta_E) diam(S)$ time units and $diam(S)+1$ asynchronous rounds follows from Theorem\ref{thm:att-selfStab}. Any subsequent $S'_L$-disturbance after reaching $\gamma^+$ can only affect nodes in $S'_B \setminus S'_L$. Let $u$ be a node affected by a $S'_L$-disturbance. Algorithm~\ref{alg:bfstree} ensures that $u$ changes its parent and the fact that $u \in S'_B \setminus S'_L$ means that $u$ has a minimal neighbor $v$ that offers a path consisting of all honest nodes to the root. After at most $|N(u)|-1$ $S'_L$-disturbances affecting $u$, $u$ starts increasing the variable 
  $i_{start}$ in Algorithm~\ref{alg:bfstree}.   
 From then on, $u$ iterates over potential parents sequentially and increases $i_{start}$ with each $S'_L$-disturbance, hence it will choose $v$ as a parent after at most 
another $|N(u)|$  $S'_L$-disturbances that affect $u$. In total, at most $2|N(u)|-1$ $S'_L$-disturbances affect $u$ before it chooses $v$ or another honest neighbor with a path of minimal length to the root. Afterwards, $u$ does not change its parent as $v$ does not change its level. 
   The claim follows by summing over all nodes $u \in S'_B\setminus S'_L$.   
\end{conference}
\begin{arxiv}
\begin{proof}
The $S'_B$-stability after $\Delta_C + (\Delta_D + \Delta_E) \cdot diam(S)$ time units plus $diam(S)+1$ asynchronous rounds follows from Theorem\ref{thm:att-selfStab}. In order to have $S'_L$-stability, all nodes in $S'_B \setminus S'_L$ have to reach a stable and legitimate state. 

Let $u \in S'_B\setminus S'_L$. The proof consists of showing the following four claims: 
\begin{enumerate}
\item $u$'s $level$-value is $level_u = d^H_S(u,r)$ for any configuration after $\gamma^+$.
\item If $u$ has a parent $v$ such that $v$ is a node on a path from $u$ to the root of length $d^H_S(u,r)$ consisting of only honest nodes, then $v \in V\setminus S'_L$ and $u$ will not change its $pID$-value in any subsequent configuration if $prnt=i_{start}$.
\item $u$ will choose such a node $v$ as a parent after at most $(2|N(u)|-1)$ $S'_L$-disturbances that affect $u$, i.e., in which $u$ changes its level or parent. 
\item The maximal number of  $S'_L$-disturbances until $u$ is in a stable and legitimate state is $2deg_{sum}-|S'_B \setminus S'_L|$. 
\end{enumerate}
By definition of $S'_B$ and $S'_L$, $u$ has at least one path consisting of only honest nodes to the root $r$. Furthermore, as $d^B_{u,min} + d^B_{min}-1=d^H_S(u,r)$, $u$ never receives a valid level attestation of length less than $d^H_S(u,r)$. So, we claim that after $diam(S)+1$ asynchronous rounds, $u$ has to have level $d^H_S(u,r)$. The previous claim obviously holds for $d^H_S(u,r)=1$ and by induction holds for all $d^H_S(u,r)$ as any honest neighbor $v$ of $u$ with  $d^H_S(v,r)=d^H_S(u,r)-1$ sends a valid level attestation to $u$.  
Hence, $u$'s $level$-value does not change and the first claim holds.

For the second claim, consider Algorithm~\ref{alg:bfstree}. $u$ always selects the minimal neighbor whose unique index is reached first. If $prnt=i_{start}$, $u$ first considers its current parent, which is $v$ (Line~\ref{line:j}). $u$ only replaces $v$ if it does not receive a valid attestation of length $d^H_S(v,r)$ and link signature from $v$.
As $v$ is honest, it does not send invalid attestations or link signatures. So, a change would only happen if $v$ changes its $level$-value. We now show that $v$ does not change its $level$-value and hence $u$ does not change its $pID$-value. 
If $v \in V\setminus S'_B$, $v$ is in a stable state and hence does not change its $level$-value. By the first part of the proof, $v \in S'_B\setminus S'_L$ also does not change its $level$-value. So, it remains to show that $v \notin S'_L$. By definition, $v$ has a path to $r$ consisting of only honest nodes and being of length $d^H_S(v,r)=d^H_S(u,r)-1$. Similarly, as $v$ is a neighbor of $u$, we have $d^B_{v,min} \geq d^B_{u,min} -1$, i.e., $v$ is at most 1 hop closer to any malicious node than $u$.
So, $d^B_{v,min} + d^B_{min}-1 \geq d^B_{u,min} + d^B_{min}-1 -1 \geq d^H_S(u,r)-1=d^H_S(v,r)$. The third step follows from Eq.~\ref{form:SL} because $u \in V\setminus S'_L$. So, $d^B_{v,min} + d^B_{min}-1 \geq d^H_S(v,r)$ and hence again by Eq.~\ref{form:SL}, $v \notin S'_L$. So, indeed, $u$ does not change its $pID$-value. 

The third claim ascertains that $u$ chooses such a $v$ as parent after at most $(2|N(u)|-1)$ $S'_L$-disturbances affecting $u$. By the above, a $S'_L$-disturbance can only affect $u$'s $pID$-value. We determine an upper bound on the number of times the $pID$-value can change until $i_{start}=prnt$ and $v$ is the parent node. Let $l$ be the local index of $v$ assigned by $u$ and 
\begin{align*}
h(m,i) = \begin{cases} 
m - i \textnormal{ if } m \geq i \\
m - i + |N(u)| \textnormal{ if } m < i
\end{cases}
\end{align*} 
As the result of a $S'_L$-disturbance, $u$'s parent changes to either $v$ or a node with pointer $prnt' \neq prnt$ with $h(prnt', i_{start})< h(l, i_{start})$. If it changes to $prnt'$, we either have $h(prnt', i_{start})< h(prnt, i_{start})$ or $h(prnt', i_{start}) > h(prnt, i_{start})$. In this first case, $i_{start}$ remains the same (Line~\ref{line:check}). However, the maximal number of consecutive decreases of the function $h(prnt, i_{start})$ is $h(l, i_{start})-1 \leq |N(u)|-1$. Once $i_{start}=prnt$, $h(prnt, i_{start})=0$. Any further change corresponds to the second case, as the condition in Line~\ref{line:check} will hold for any new parent and so $h(p, i_{start})$ continues to be $0$. In the second case, i.e., $h(prnt', i_{start}) > h(prnt, i_{start})$, $i_{start}$ is now set to $p'$, i.e., $h(l, i_{start})$ decreases. $h(l, i_{start})$ can decrease at most $|N(u)|$ times.
So, the total number of $S'_L$-disturbances until $u$ chooses $v$ as a parent are the sum of possible instance of the first and the second case, namely $|N(u)|-1+|N(u)|=2|N(u)|-1$. Furthermore, $i_{start}=prnt$ holds after these disturbances.

The fourth and last claim establishes that all nodes in $V\setminus S'_L$ are in a legitimate and stable state after at most $2deg_{sum}-|S'_B \setminus S'_L|$ disturbances. First note that each $S'_L$-disturbance has to affect a node in $S'_B\setminus S'_L$.  This is a direct consequence of the definition of $S'_L$-disturbance and the fact that $S$ is $S'_B$-stable. A $S'_L$-disturbance requires a node in $V \setminus S'_L$ to change its level or parent but nodes in $V \setminus S'_B$ are in a stable state already, so the affected node has to be in $(V \setminus S'_L) \cap S'_B = S'_B \setminus S'_L$. 
Combining the second and third claim, nodes $u \in S'_B \setminus S'_L$ are affected at most $2|N(u)|-1$ times by a $S'_L$-disturbance. So, the total number of $S'_L$-disturbances until no node in $V\setminus S'_L$ can be affected anymore is $\sum_{u \in S'_B \setminus S'_L} (2|N(u)|-1) = 2deg_{sum}-|S'_B \setminus S'_L|$.
It remains to show that all these nodes are indeed in legitimate states. By the second and third claim, all nodes in $S'_B \setminus S'_L$ have an honest parent in $V \setminus S'_L$. In addition, all nodes in $V \setminus S'_B$ have an honest parent in $V \setminus S'_L$ because of the $S'_B$-stability.
Hence, a node $u \in S'_B\setminus S'_L$ cannot have an ancestor in $B \cup S'_L$ and is hence well-directed and in a legitimate state.  Furthermore, $u$ is in a stable state by the second claim. 
 \end{proof}
\end{arxiv}

\section{Evaluation}
\label{sec:evaluation}

\noindent Using OMNeT++~\cite{omnetpp}, we implemented a simulation to evaluate the impact of our attestation-based algorithm on the number of lost nodes compared to the non-cryptographic state-of-the-art. 
Furthermore, we investigated the impact of the network structure, the position of the root node, and the placement of edges between honest and malicious nodes.

\subsection{Metrics, Data Sets, and System Parameters}
\noindent Given a distributed system $S=(V,E)$ with a subset $H$ of honest nodes and a $S_B$-TA stabilizing spanning tree construction algorithm, we measured the ratio of lost nodes (RLN) $\nicefrac{|S_B|}{|H|}$.
A low ratio of lost nodes indicates high attack resistance.
\begin{table}
\small
   \renewcommand{\arraystretch}{1.3}
   \setlength\tabcolsep{1mm}
   \caption{Structural properties of graphs used for simulation, with avg. shortest path length (CPL) and clustering coefficient (CC).}
   \label{tab:graphs}
   \centering
   \begin{tabular}{|c|c|c|c|c|}
      \hline
      \textbf {Graph} & \textbf {\# nodes} & \textbf {\# edges} & \textbf {CPL} & \textbf{CC}\\ \hline
      Facebook & 63,392 & 816,886 & 4.32 & 0.253 \\
      Ripple & 67,149 & 99,787 & 3.82 & 0.154 \\ \hline
      Randomized Facebook & \multirow{2}{*}{63,392} & 816,886 & 3.58 & 0.005\\ 
      Erdös-Renyi & & 824,096 & 3.74 & < 0.001\\ \hline
   \end{tabular}
   \vspace{-1.7em}
\normalsize
\end{table}

\begin{figure*}
   \centering
   \subfloat{
\begin{tikzpicture}[x=1pt,y=1pt]
\definecolor{fillColor}{RGB}{255,255,255}
\begin{scope}
\definecolor{drawColor}{RGB}{255,0,0}

\path[draw=drawColor,line width= 0.8pt,line join=round,line cap=round] ( 13.38, 13.51) --
	( 15.81,  9.32) --
	( 10.96,  9.32) --
	( 13.38, 13.51);
\end{scope}
\begin{scope}
\definecolor{drawColor}{RGB}{0,0,0}

\node[text=drawColor,anchor=base,inner sep=0pt, outer sep=0pt, scale=  0.80] at ( 51.02,  8.93) {without attestation};
\definecolor{drawColor}{RGB}{0,205,0}

\path[draw=drawColor,line width= 0.8pt,line join=round,line cap=round] (111.21, 10.84) -- (116.30, 10.84);

\path[draw=drawColor,line width= 0.8pt,line join=round,line cap=round] (113.76,  8.29) -- (113.76, 13.39);
\definecolor{drawColor}{RGB}{0,0,0}

\node[text=drawColor,anchor=base,inner sep=0pt, outer sep=0pt, scale=  0.80] at (147.22,  8.93) {with attestation};
\definecolor{drawColor}{RGB}{0,0,255}

\path[draw=drawColor,line width= 0.8pt,line join=round,line cap=round] (203.97,  9.04) -- (207.57, 12.64);

\path[draw=drawColor,line width= 0.8pt,line join=round,line cap=round] (203.97, 12.64) -- (207.57,  9.04);
\definecolor{drawColor}{RGB}{0,0,0}

\node[text=drawColor,anchor=base,inner sep=0pt, outer sep=0pt, scale=  0.80] at (260.14,  8.99) {honestly behaving adversary};
\end{scope}
\end{tikzpicture}%
   } \vspace{-0.8em} \\
   \subfloat{
      \input{results/means-facebook-wosn-true.tex}%
   }
   \subfloat{
      \input{results/means-facebook-random-true.tex}
   }
   \subfloat{
      \input{results/means-ER_63392_26_0-true.tex}
   }
   \subfloat{
      \input{results/means-ripple-lcc-true.tex}
   }
   \caption{Observed mean ratio of lost nodes over 100 runs per configuration for 25, 200, 1000, and 5000 attack edges under the first adversarial behavior. The bars above and below each point represent $99\%$ confidence intervals.}
   \label{fig:means}
   \vspace{-0.8em}
\end{figure*}
%
%
\begin{arxiv}
   \begin{figure*}
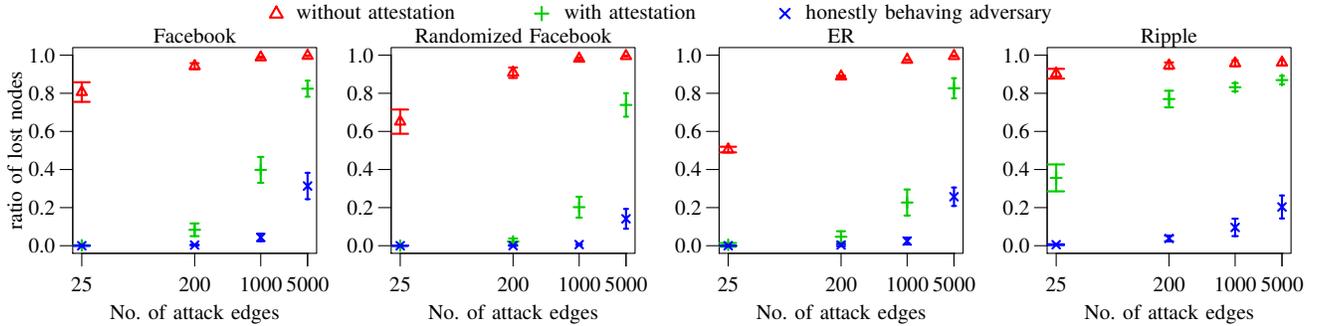

      \centering
      \subfloat{
\begin{tikzpicture}[x=1pt,y=1pt]
\definecolor{fillColor}{RGB}{255,255,255}
\begin{scope}
\definecolor{drawColor}{RGB}{255,0,0}

\path[draw=drawColor,line width= 0.8pt,line join=round,line cap=round] ( 13.38, 13.51) --
	( 15.81,  9.32) --
	( 10.96,  9.32) --
	( 13.38, 13.51);
\end{scope}
\begin{scope}
\definecolor{drawColor}{RGB}{0,0,0}

\node[text=drawColor,anchor=base,inner sep=0pt, outer sep=0pt, scale=  0.80] at ( 51.02,  8.93) {without attestation};
\definecolor{drawColor}{RGB}{0,205,0}

\path[draw=drawColor,line width= 0.8pt,line join=round,line cap=round] (111.21, 10.84) -- (116.30, 10.84);

\path[draw=drawColor,line width= 0.8pt,line join=round,line cap=round] (113.76,  8.29) -- (113.76, 13.39);
\definecolor{drawColor}{RGB}{0,0,0}

\node[text=drawColor,anchor=base,inner sep=0pt, outer sep=0pt, scale=  0.80] at (147.22,  8.93) {with attestation};
\definecolor{drawColor}{RGB}{0,0,255}

\path[draw=drawColor,line width= 0.8pt,line join=round,line cap=round] (203.97,  9.04) -- (207.57, 12.64);

\path[draw=drawColor,line width= 0.8pt,line join=round,line cap=round] (203.97, 12.64) -- (207.57,  9.04);
\definecolor{drawColor}{RGB}{0,0,0}

\node[text=drawColor,anchor=base,inner sep=0pt, outer sep=0pt, scale=  0.80] at (260.14,  8.99) {honestly behaving adversary};
\end{scope}
\end{tikzpicture}%
      } \vspace{-0.8em} \\
      \subfloat{
         \input{results/means-facebook-wosn-false.tex}%
      }
      \subfloat{
         \input{results/means-facebook-random-false.tex}
      }
      \subfloat{
         \input{results/means-ER_63392_26_0-false.tex}
      }
      \subfloat{
         \input{results/means-ripple-lcc-false.tex}
      }
      \caption{Observed mean ratio of lost nodes over 100 runs per configuration for 25, 200, 1000, and 5000 attack edges under the second adversarial behavior. The bars above and below each point represent $99\%$ confidence intervals.}
      \label{fig:means-illdirected}
      \vspace{-1.7em}
   \end{figure*}
\end{arxiv}

Route-restricted overlays include both social overlays and payment networks. 
We hence utilized a real-world graph for each of them and compare the results with synthetic graphs for the purpose of characterizing the impact of various topological features. 
\emph{Facebook} denotes a real-world graph of Facebook~\cite{viswanath2009evolution}, as used in several prior studies~\cite{mittal2012,roos2016anonymous}.
Ripple denotes a real-world graph from the Ripple payment network~\cite{roos2018settling}. Ripple has a low number of edges and a heavily skewed degree distribution: $95\%$ of all nodes have a degree less or equal than the average degree of approximately 3.
Our synthetic data sets are  i) a random synthetic network (denoted \emph{randomized Facebook}) with the same degree distribution as the Facebook graph and ii) an Erdös and Renyi graph (\emph{ER}) with approximately the same number of nodes and edges as \emph{Facebook} but normal distributed degrees~\cite{erdos1959random}. We compare \emph{Facebook} with \emph{randomized Facebook} to characterize the impact of clustering while the comparison of \emph{randomized Facebook} and \emph{ER} reveals the impact of the degree distribution. 

We considered the number of malicious nodes and the time of their presence to be unbounded but limit the total number $g$ of connections between honest nodes and malicious nodes.
To model that all nodes are colluding, we represented them as a single node with $g$ edges.

\subsection{Set-up}
\noindent We investigated the resistance of spanning tree algorithms to adversarial behavior given structural differences of the networks and a varying number $g$ of attack edges. For all scenarios, we performed 100 runs to obtain statistically significant results.

We assumed the adversary knows all nodes but can only establish a connection to a subset with limited size.
Following~\cite{boshmaf2011socialbot} we also assumed that users with many contacts are more likely to accept new requests and thus connect with a malicious node.
We hence added a single adversary $m$ to the graph and choose the $g$ honest neighbors at random, with a probability proportional to their degree.
Afterwards, a root node $r$ was chosen uniformly at random from all honest nodes and the leader ID of each honest node was set accordingly.

We executed different spanning tree constructions for various adversarial behaviors. 
The two spanning tree algorithms are Algorithm~\ref{alg:bfstree}, i.e., spanning tree construction with level attestation, and the state-of-the-art protocol by Dubois et al.~\cite{dubois2015maximum}.
The two adversarial behaviors are:
\begin{enumerate}
\item The attacker aims to prevent convergence by causing disturbances. By Theorem~\ref{thm:SL}, the set of lost nodes corresponds to $S'_L$ as defined in Eq.~\ref{form:SL}. Similarly, the set of lost nodes for the state-of-the-art protocol is $S_L = \{u \in H: d_S(u,m) < d_S(u,r)\}$~\cite{dubois2015maximum}.
\item The attacker aims to maximize the number of ill-directed nodes. In this case, the adversary always pretends to be as close to the root as possible and does not perform any disturbances. In this case, the set of lost nodes is $S'_B$ as defined in Eq.~\ref{form:SB} according to Theorem~\ref{thm:att-selfStab}. For the state-of-the-art protocol, the set of lost nodes is $S_B = \{u \in H: d_S(u,m) < d_S(u,r)\}$.
\end{enumerate}
To investigate how strongly the cheating by one level (described in Sec.~\ref{sec:proofs}) affects the number of lost nodes when Algorithm~\ref{alg:bfstree} is used, we furthermore simulated a modified adversary which does not cheat, effectively following Algorithm~\ref{alg:bfstree} correctly.
\begin{conference}
As the results for the second adversarial behavior are very similar in terms of the advantage gained by Algorithm~\ref{alg:bfstree}, we only present the results for the first adversarial behavior here. The remaining results are included in the extended version~\cite{byrenheid2019attack}.
\end{conference}

\begin{figure*}
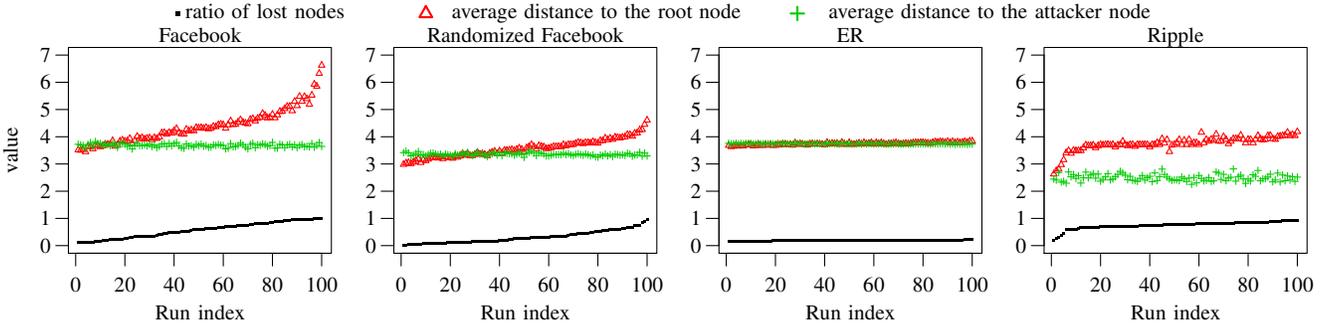

   \centering
   \subfloat{
\begin{tikzpicture}[x=1pt,y=1pt]
\definecolor{fillColor}{RGB}{255,255,255}
\begin{scope}
\definecolor{fillColor}{RGB}{0,0,0}

\path[fill=fillColor] ( 17.87,  9.85) rectangle ( 19.60, 11.58);
\end{scope}
\begin{scope}
\definecolor{drawColor}{RGB}{0,0,0}

\node[text=drawColor,anchor=base,inner sep=0pt, outer sep=0pt, scale=  0.80] at ( 51.53,  8.93) {ratio of lost nodes};
\definecolor{drawColor}{RGB}{255,0,0}

\path[draw=drawColor,line width= 0.8pt,line join=round,line cap=round] (112.42, 13.64) --
	(114.84,  9.44) --
	(110.00,  9.44) --
	(112.42, 13.64);
\definecolor{drawColor}{RGB}{0,0,0}

\node[text=drawColor,anchor=base,inner sep=0pt, outer sep=0pt, scale=  0.80] at (176.83,  8.99) {average distance to the root node};
\definecolor{drawColor}{RGB}{0,205,0}

\path[draw=drawColor,line width= 0.8pt,line join=round,line cap=round] (250.40, 10.84) -- (255.49, 10.84);

\path[draw=drawColor,line width= 0.8pt,line join=round,line cap=round] (252.94,  8.29) -- (252.94, 13.39);
\definecolor{drawColor}{RGB}{0,0,0}

\node[text=drawColor,anchor=base,inner sep=0pt, outer sep=0pt, scale=  0.80] at (325.55,  8.99) {average distance to the attacker node};
\end{scope}
\end{tikzpicture}%
   } \vspace{-0.8em} \\
   \subfloat{
      \input{results/values-trad-facebook-wosn.tex}
   }
   \subfloat{
      \input{results/values-trad-facebook-random.tex}
   }
   \subfloat{
      \input{results/values-trad-ER_63392_26_0.tex}
   }
   \subfloat{
      \input{results/values-trad-ripple-lcc.tex}
   }
   \caption{Obtained RLN values for the state-of-the-art protocol together with the average shortest path length to the root node and average shortest path length to the adversary node of the respective simulation run with 25 attack edges. The runs are ordered according to the RLN value in ascending order.}
   \label{fig:results-traditional}
      \vspace{-0.5em}
\end{figure*}
\begin{figure*}
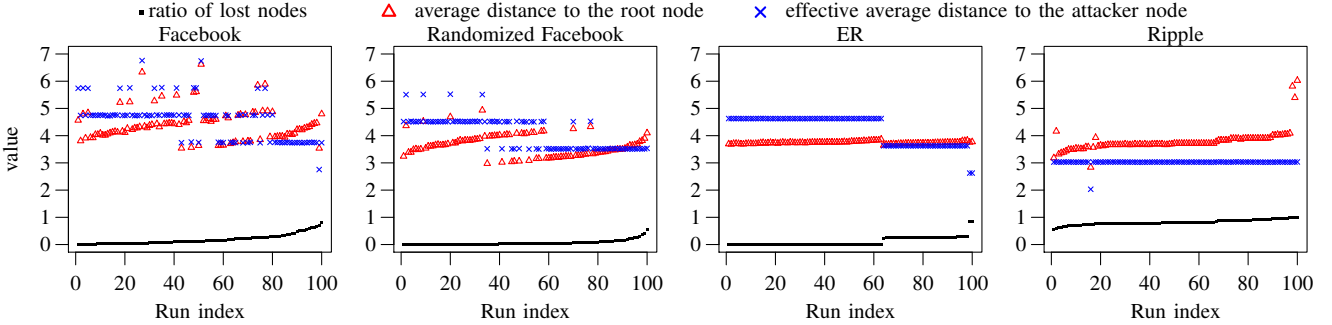

   \centering
   \subfloat{
\begin{tikzpicture}[x=1pt,y=1pt]
\definecolor{fillColor}{RGB}{255,255,255}
\begin{scope}
\definecolor{fillColor}{RGB}{0,0,0}

\path[fill=fillColor] ( 17.87,  9.85) rectangle ( 19.60, 11.58);
\end{scope}
\begin{scope}
\definecolor{drawColor}{RGB}{0,0,0}

\node[text=drawColor,anchor=base,inner sep=0pt, outer sep=0pt, scale=  0.80] at ( 51.53,  8.93) {ratio of lost nodes};
\definecolor{drawColor}{RGB}{255,0,0}

\path[draw=drawColor,line width= 0.8pt,line join=round,line cap=round] (112.42, 13.64) --
	(114.84,  9.44) --
	(110.00,  9.44) --
	(112.42, 13.64);
\definecolor{drawColor}{RGB}{0,0,0}

\node[text=drawColor,anchor=base,inner sep=0pt, outer sep=0pt, scale=  0.80] at (176.83,  8.99) {average distance to the root node};
\definecolor{drawColor}{RGB}{0,0,255}

\path[draw=drawColor,line width= 0.8pt,line join=round,line cap=round] (251.14,  9.04) -- (254.75, 12.64);

\path[draw=drawColor,line width= 0.8pt,line join=round,line cap=round] (251.14, 12.64) -- (254.75,  9.04);
\definecolor{drawColor}{RGB}{0,0,0}

\node[text=drawColor,anchor=base,inner sep=0pt, outer sep=0pt, scale=  0.80] at (338.43,  8.99) {effective average distance to the attacker node};
\end{scope}
\end{tikzpicture}%
   } \vspace{-0.8em} \\
   \subfloat{
      \input{results/values-att-facebook-wosn.tex}
   }
   \subfloat{
      \input{results/values-att-facebook-random.tex}
   }
   \subfloat{
      \input{results/values-att-ER_63392_26_0.tex}
   }
   \subfloat{
      \input{results/values-att-ripple-lcc.tex}
   }
   \caption{Results for the simulation runs of Algorithm~\ref{alg:bfstree} and 1000 attack edges. The \emph{effective average distance} to $m$ denotes the term $\overline{d}(m)+d(m,r)-1$. The runs are ordered according to the RLN value in ascending order.}
   \label{fig:results-attestation}
   \vspace{-1.7em}
\end{figure*}
 
\subsection{Impact of Level Attestation}
\noindent 
\begin{conference}
\figurename~\ref{fig:means} show the obtained mean RLN with $99\%$ confidence intervals for the four graphs and both algorithms under the first adversarial behavior.
Especially for the Facebook graph, its randomized version, and the ER graph, Algorithm~\ref{alg:bfstree} considerably reduced the ratio of lost nodes compared to the state-of-the-art protocol.
For the latter, an attack with 25 edges resulted in a mean RLN of $0.57$, $0.31$, $0.19$, and $0.75$ for the Facebook graph, the randomized Facebook graph, the ER graph and the Ripple graph, respectively.
When applying Algorithm~\ref{alg:bfstree}, the mean RLN at 25 attack edges dropped down to $0.0005$, $0.00006$, $0.00005$, and $0.34$ for the Facebook graph, the randomized Facebook graph, the ER graph, and the Ripple graph, respectively.
Even for 1000 attack edges, the mean RLN for the Facebook graph, its randomized variant, and the ER graph significantly decreased from $0.93$, $0.84$, and $0.84$ to $0.18$, $0.06$, and $0.12$, respectively.

In the scenario with an adversary that does not cheat by a level, the mean RLN was considerably lower than in the scenario with Algorithm~\ref{alg:bfstree} alone, especially with 1000 and 5000 attack edges. 
Referring to Table \ref{tab:graphs} we realize that all graphs used in our experiment have a very low average path length, and all nodes are in short distance from the root node.
Increasing the adversary's reported $level$ value by 1 then represents a significant disadvantage for the attack.
We conjecture that this causes many nodes to remain well-directed and investigate this relationship in more detail in the following.

For the Ripple graph, the improvement regarding mean RLN was considerably lower than for the other graphs.
In the following, we describe the impact of distances between honest nodes, malicious nodes, and the root node on the RLN to explain this stark difference. 
\end{conference}

\begin{arxiv}
\figurename~\ref{fig:means} show the obtained mean RLN with $99\%$ confidence intervals for the four graphs and both algorithms under the first adversarial behavior.
\figurename~\ref{fig:means-illdirected} shows the corresponding data under the second adversarial behavior.
Especially for the Facebook graph, its randomized version, and the ER graph, Algorithm~\ref{alg:bfstree} considerably reduced the ratio of lost nodes compared to the state-of-the-art protocol.
For the latter, an attack with 25 edges resulted in a mean RLN of $0.57$, $0.31$, $0.19$, and $0.75$ for the Facebook graph, the randomized Facebook graph, the ER graph and the Ripple graph, respectively, under the first adversarial behavior.
Under the second adversarial behavior, the mean RLN increased to $0.77$, $0.66$, $0.5$ and $0.91$ for the four graphs.
When applying Algorithm~\ref{alg:bfstree}, the mean RLN at 25 attack edges under the first adversarial behavior dropped down to $0.0005$, $0.00006$, $0.00005$, and $0.34$ for the Facebook graph, the randomized Facebook graph, the ER graph, and the Ripple graph, respectively.
Similarly, the mean RLN at 25 attack edges under the second adversarial behavior decreased to $0.002$, $0.0006$, $0.004$ and $0.38$ for the four graphs.
Even for 1000 attack edges, the mean RLN for the Facebook graph, its randomized variant, and the ER graph significantly decreased from $0.93$, $0.84$, and $0.84$ to $0.18$, $0.06$, and $0.12$, respectively under the first adversarial behavior.
Under the second adversarial behavior, the mean RLN at 1000 attack edges was reduced from $0.98$, $0.97$, $0.98$ and $0.95$ to $0.36$, $0.21$ and $0.23$ for the Facebook graphs and the ER graph.
In summary, while the exact numbers differ for the two adversarial behaviors, the overall result is the same: Algorithm~\ref{alg:bfstree} achieves a considerable higher number of well directed nodes than the state of the art.    

In the scenario with an adversary that does not cheat by a level, the mean RLN was considerably lower than in the scenario with Algorithm~\ref{alg:bfstree} alone, especially with 1000 and 5000 attack edges for both adversarial behaviors.
Referring to Table \ref{tab:graphs} we realize that all graphs used in our experiment have a very low average path length, and all nodes are in short distance from the root node.
Increasing the adversary's reported $level$ value by 1 then represents a significant disadvantage for the attack.
We conjecture that this causes many nodes to remain well-directed and investigate this relationship in more detail in the following.

For the Ripple graph, the improvement regarding mean RLN was considerably lower than for the other graphs.
In the following, we describe the impact of distances between honest nodes, malicious nodes, and the root node on the RLN to explain this stark difference. 
\end{arxiv}

\subsection{Impact of Network Structure}
We start with a discussion of our results for the state-of-the-art algorithm and subsequently present results for Algorithm~\ref{alg:bfstree}.
Because the correlations between the different aspects were similar for both adversarial behaviors, we focus on our results for the first adversarial behavior.

\paragraph{State-of-the-Art Spanning Tree Construction}
We considered the average hop distance over all honest nodes to the root node $\overline{d}(r)$ and to the attacker node $\overline{d}(m)$ for each simulation run.
The lower $\overline{d}(r)$ is compared to $\overline{d}(m)$, the more honest nodes will have a lower hop distance to $r$ than to $m$ and thus be well-directed. Therefore, we expected a positive correlation between $\overline{d}(r)-\overline{d}(m)$ and the RLN. 

\figurename~\ref{fig:results-traditional} shows the obtained RLN values in ascending order together with the corresponding value of $\overline{d}(r)$ and $\overline{d}(m)$ for 25 attack edges and both adversarial behaviors. 
Indeed, the difference $\overline{d}(r)-\overline{d}(m)$ generally correlated with the RLN. 
While $\overline{d}(m)$ only varied slightly between the different runs on each graph, there are notable differences in the behavior of $\overline{d}(r)$: It varied highly for the Facebook graphs and to some extent for the Ripple graph but barely for the ER graph.  

The reason for the small variance in $\overline{d}(r)$ for ER is due to the uniform probability of two nodes being connected. 
As a consequence, it was very unlikely that the average distance of any node significantly differs from the other nodes. 
The degree of $m$, corresponding to the 25 attack edges, was close to the average degree of 26.
However, the mean RLN was only $0.19$, because there was a high number of nodes whose distance to the root node equalled that to the malicious node. These nodes chose the path to the root node when the malicious node continuously causes disturbances. 
\begin{conference}
In the setting where the attacker aims to maximize the number of ill-directed nodes, approximately half of all nodes became ill-directed.
\end{conference}

For the Facebook graphs and Ripple, there is a higher variance of the root node degree and hence of the average distance to the root. The distances in the randomized graph were generally lower than in the original Facebook graph due to its lower average path length. Furthermore, $\overline{d}(r)$ correlated more strongly with the RLN, possibly due to the absence of outlier nodes with increased shortest path length.
Because of the highly skewed degree distribution of the Ripple graph, the random root node's degree was $1$ in 73 out of a 100 runs. The degree of the adversarial node, i.e., 25, was hence generally higher than the degree of the root, leading to shorter paths to the malicious nodes and hence the observed high RLN. 

\paragraph{Algorithm~\ref{alg:bfstree}}
In addition to $\overline{d}(r)$, we computed the effective average hop distance $D(m)=\overline{d}(m)+d(m,r)-1$, as $d(m,r)-1$ is the $level$-value that $m$ propagated during a simulation run. We expected a positive correlation between $\overline{d}(r)-D(m)$ and the RLN, i.e., nodes closer to the root node than $D(m)$ should be well-directed and otherwise not. 

In all runs with 25 attack edges, $D(m)$ was considerably higher than $\overline{d}(r)$ such that only a very small number of nodes became ill-directed.
\figurename~\ref{fig:results-attestation} thus shows our more distinct results for an adversary with 1000 attack edges, ordered by the RLN. The results indeed validated the expected correlation.  
Due to the high number of attack edges, the $\overline{d}(m)$ value of each run only differed slightly from its mean value of $2.75$, $2.52$, $2.63$, and $2.02$ for the Facebook graph, the randomized Facebook graph, the ER graph, and the Ripple graph, respectively. 
Thus, the values of $D(m)$ mainly depended on $d(m,r)$ and hence differed by integer values.

Again, the degree of correlation between $\overline{d}(r)-D(m)$ and the RLN varied between graphs. The Facebook graph generally had a longer average shortest path length and hence varied in $\overline{d}(r)$ considerably. In contrast, the value of $\overline{d}(r)$ was more stable for the randomized Facebook graph and the ER graph, so that $d(m,r)$ is indeed the main impact factor.  

Here, we also find the explanation for the strong difference between the mean RLN values for the simulations of Algorithm~\ref{alg:bfstree} with a cheating adversary and those of Algorithm~\ref{alg:bfstree} with a non-cheating adversary on the Ripple graph. It stems from the fact that the $d(m,r)$ value decreased very slowly as the number of attack edges increases. Concretely, the mean value of $\overline{d}(r)$ was roughly $3.8$, irrespective of the number of attack edges and the construction algorithm.
In the case of 25 attack edges, the mean value for $D(m)$ was $3.9$ and in the case of $5000$ edges, it was $2.9$, such that the $level$ value propagated by $m$ was low enough to cause a high number of nodes to become ill-directed.
As the value of $D(m)$ was increased by 1 when the adversary does not cheat, it was higher than $\overline{d}(r)$ for any considered number of attack edges, resulting in a negative $\overline{d}(r)-D(m)-1$ and hence a low impact of the attack. In contrast, $\overline{d}(r)-D(m)$ was positive, corresponding to an attack of high impact.  

\paragraph{Summary of Results}
The first part of our evaluation showed that our protocol based on cryptographic signatures is much more robust to malicious behavior and attacks than state-of-the-art solutions without the usage of cryptography.
Indeed, as displayed in \figurename~\ref{fig:means}, to compromise a similar number of nodes, the adversary needs to establish up to 200 times as many attack edges compared to the algorithm by Dubois et al.~\cite{dubois2015maximum}.

\section{Conclusion}
\label{sec:conclusion}
In this paper, we leveraged cryptographic signatures to design a BFS tree algorithm that greatly reduces the number of nodes affected by attacks.
Based on the concept of topology-aware strict stabilization, we proved that this algorithm only allows malicious nodes to report a distance to the root that differs by at most one from the correct value. 
Our evaluation based on real-world scenarios demonstrates that this novel construction provides crucial security improvements over existing, non-cryptographic algorithms.
Yet, our results indicate that the resistance to attacks is highly correlated with the degree of the root node, highlighting the need to develop secure leader election algorithms that prioritize high-degree nodes.

\section{Acknowledgements}
We thank Sebasti\'{e}n Tixeuil for shepherding our work and the reviewers for their constructive feedback.
This work has been funded by the German Research Foundation (DFG) Grant STR 1131/2-2 and EXC 2050 ``CeTI´´.
\bibliographystyle{amsplain}
\bibliography{articles}

%
%

\end{document}